\documentclass[a4paper,12pt]{amsart}

\usepackage{amsmath,amssymb,amsthm}

\usepackage[foot]{amsaddr}
\usepackage{graphicx}
\usepackage{hyphenat}
\usepackage{tikz}
\usepackage{todonotes}
\usepackage[utf8]{inputenc}
\usepackage[T1]{fontenc}

\newtheorem{theorem}{Theorem}

\let\leq\leqslant
\let\geq\geqslant

\newcommand{\set}[1]{\left\{#1\right\}}

\newcommand{\cnt}[1]{\#\mathsf{#1}}

\address{Theoretical Computer Science Department, Faculty of Mathematics and Computer Science, Jagiellonian University, Krak\'{o}w, Poland}

\begin{document}

\thanks{
  Grzegorz Guśpiel was partially supported by the MNiSW grant DI2013 000443.
}

\title[]{Complexity of Finding \\ Perfect Bipartite Matchings Minimizing the Number of Intersecting Edges}
\author{Grzegorz Guśpiel}

\begin{abstract}
Consider a problem where we are given a~bipartite graph $H$
with vertices arranged on two horizontal lines in the plane,
such that the two sets of vertices placed on the two lines
form a~bipartition of $H$.
We additionally require that $H$ admits a~perfect matching
and assume that edges of $H$ are embedded in the plane as segments.
The goal is to compute the minimal number of intersecting edges
in a perfect matching in $H$.

The problem stems from so-called \emph{token swapping} problems,
  introduced by Yamanaka~et~al. \cite{Yamanaka2015} and
  generalized by Bonnet, Miltzow and Rzążewski \cite{BonnetMR17}.
We show that our problem,
equivalent to one of the special cases of one of the token swapping problems,
is $\mathbf{NP}$-complete.
\end{abstract}

\maketitle

\section{Preliminaries}

All graphs considered here are undirected, loopless and without multiple edges.
For a graph $G$, the sets $V(G), E(G)$ are the set of vertices and the set of edges of $G$, respectively.
For a vertex $v$, the number $d(v)$ is its degree.

\section{Token swapping problems}

In 2014, Yamanaka~et~al.~\cite{Yamanaka2014,Yamanaka2015} introduced \textsc{Token Swapping}
-- a~problem, in which
we are given a graph, whose every vertex contains a token.
Every token has a destination vertex.
In one turn, it is allowed to swap tokens on adjacent vertices.
The goal is to bring every token to its destination vertex in the smallest
number of turns.

Two generalizations of the problem were introduced.
In the first one, 
presented in the paper by Yamanaka~et~al. and
called \textsc{Colored Token Swapping},
tokens and vertices have colors and
the goal is to move each token to a vertex of the same color.
In the second generalization, 
introduced by Bonnet, Miltzow and Rzążewski \cite{BonnetMR17}
and called \textsc{Subset Token Swapping},
every token is assigned 
a set of destination vertices and is required to reach any of them.
For a more precise definition, we refer to the original paper.

\textsc{Token Swapping} 
was shown to be $\mathbf{NP}$-complete
by Miltzow~et~al.\cite{MiltzowNORTU16}.
This means that the two generalizations are $\mathbf{NP}$-complete as well.
One of the results of Bonnet~et~al.~is an analysis of the complexity of 
all three token swapping problems
on certain simple graph classes,
i.e.~trees, cliques, stars and paths.
\textsc{Subset Token Swapping}
was shown to be $\mathbf{NP}$-complete on the first three classes,
but the status of the problem for paths was unknown.
We complement this analysis by proving that
\textsc{Subset Token Swapping} restricted to paths
is $\mathbf{NP}$-complete as well.

\section{Hardness proof}

Miltzow \cite{Miltzow16} showed that
\textsc{Subset Token Swapping} restricted to paths
is equivalent to our matching problem, which we now formally define as follows:

\medskip
\noindent \underline{\textsc{Crossing-avoiding Matching}}:\smallskip \\
\noindent{\bf Input:}
$(H, V_1, V_2, \sigma_1, \sigma_2, m)$, 
where $H$ is a bipartite graph that admits at least one perfect matching,
$(V_1, V_2)$ is a bipartition of $H$,
$\sigma_1$ is a permutation of $V_1$,
$\sigma_2$ is a permutation of $V_2$,
and $m$ is a nonnegative integer.
We interpret the permutations as a placement of $V(H)$
on two different parallel horizontal lines in the plane:
vertices in $V_1$ ($V_2$) are assigned different points in the upper (lower) line 
and $\sigma_1$ ($\sigma_2$) is their order from left to right.
Edges are embedded in the plane as segments.

\noindent{\bf Question:}
Does $H$ admit a perfect matching such that at most $m$ pairs of edges intersect? \medskip

We prove the $\mathbf{NP}$-completenes of
\textsc{Subset Token Swapping} on paths 
by showing the following:

\begin{theorem}
$\textsc{Crossing-avoiding Matching}$ is 
$\mathbf{NP}$-complete,
even if the maximum degree of $H$ is 2.
\end{theorem}
\begin{proof}
  It is trivial that the problem is in $\mathbf{NP}$.
  Our proof of $\mathbf{NP}$-hardness is a polynomial-time reduction
  from $\textsc{Vertex Cover}$.

  The reduction algorithm takes as input a \textsc{Vertex Cover} instance $(G, k)$
  and outputs an instance $(H, V_1, V_2, \sigma_1, \sigma_2, m)$
  of \textsc{Crossing-avoiding Matching}
  that is a \emph{YES}-instance if and only if 
  $G$ admits a vertex cover of size at most $k$.
  Our construction is based on two gadgets.
  The first one is created for every vertex of $G$. For every integer $s \geq 1$, the \emph{vertex gadget of size $s$}
  is a cycle on $4s$ vertices together with a path on 2 vertices, positioned 
  as shown in Figure~\ref{vg}.
  \begin{figure}[h]
    \begin{center}
      \begin{tikzpicture}
\begin{scope}[xscale=0.55, yscale=0.55]
\definecolor{cborange}{RGB}{230,159,0}
\definecolor{cbblue}{RGB}{0,114,178}
\definecolor{cbpurple}{RGB}{204,121,167}
\definecolor{cbyellow}{RGB}{240,228,66}
\definecolor{cbgreen}{RGB}{0,158,115}
\definecolor{cbvermillon}{RGB}{213,94,0}
\definecolor{mgray}{RGB}{220,220,220}
\definecolor{mlightblue}{RGB}{220,255,255}
\tikzstyle{every node}=[circle,minimum size=3pt,inner sep=0pt,draw,fill]
\tikzstyle{every path}=[line width = 1.9pt]
\scriptsize
\node (r1) at (1, 0) {};
\node (rp1) at (1, 2) {};
\node (l1) at (-1, 0) {};
\node (lp1) at (-1, 2) {};
\node (r2) at (3, 0) {};
\node (rp2) at (3, 2) {};
\node (l2) at (-3, 0) {};
\node (lp2) at (-3, 2) {};
\node (r3) at (5, 0) {};
\node (rp3) at (5, 2) {};
\node (l3) at (-5, 0) {};
\node (lp3) at (-5, 2) {};
\node (r4) at (7, 0) {};
\node (rp4) at (7, 2) {};
\node (l4) at (-7, 0) {};
\node (lp4) at (-7, 2) {};
\node (r5) at (9, 0) {};
\node (rp5) at (9, 2) {};
\node (l5) at (-9, 0) {};
\node (lp5) at (-9, 2) {};
\node (r6) at (11, 0) {};
\node (rp6) at (11, 2) {};
\node (l6) at (-11, 0) {};
\node (lp6) at (-11, 2) {};
\filldraw[fill=mgray, draw=black,thin,dashed] (1.8, -0.3) rectangle (2.2, 2.3);
\filldraw[fill=mgray, draw=black,thin,dashed] (-1.8, -0.3) rectangle (-2.2, 2.3);
\node [rectangle, fill=none, draw=none, anchor=south, align=center, color = gray, text width = 2cm] at (2, -0.9) {3};
\node [rectangle, fill=none, draw=none, anchor=south, align=center, color = gray, text width = 2cm] at (-2, -0.9) {3};
\filldraw[fill=mgray, draw=black,thin,dashed] (5.8, -0.3) rectangle (6.2, 2.3);
\filldraw[fill=mgray, draw=black,thin,dashed] (-5.8, -0.3) rectangle (-6.2, 2.3);
\node [rectangle, fill=none, draw=none, anchor=south, align=center, color = gray, text width = 2cm] at (6, -0.9) {2};
\node [rectangle, fill=none, draw=none, anchor=south, align=center, color = gray, text width = 2cm] at (-6, -0.9) {2};
\filldraw[fill=mgray, draw=black,thin,dashed] (9.8, -0.3) rectangle (10.2, 2.3);
\filldraw[fill=mgray, draw=black,thin,dashed] (-9.8, -0.3) rectangle (-10.2, 2.3);
\node [rectangle, fill=none, draw=none, anchor=south, align=center, color = gray, text width = 2cm] at (10, -0.9) {1};
\node [rectangle, fill=none, draw=none, anchor=south, align=center, color = gray, text width = 2cm] at (-10, -0.9) {1};
\draw [color=cborange](r1)--(rp2);
\draw [color=cborange](rp1)--(r2);
\draw [color=cborange](l1)--(lp2);
\draw [color=cborange](lp1)--(l2);
\draw [color=cborange](r3)--(rp4);
\draw [color=cborange](rp3)--(r4);
\draw [color=cborange](l3)--(lp4);
\draw [color=cborange](lp3)--(l4);
\draw [color=cborange](r5)--(rp6);
\draw [color=cborange](rp5)--(r6);
\draw [color=cborange](l5)--(lp6);
\draw [color=cborange](lp5)--(l6);
\draw [color=cbblue](r2)--(rp3);
\draw [color=cbblue](r3)--(rp2);
\draw [color=cbblue](l2)--(lp3);
\draw [color=cbblue](l3)--(lp2);
\draw [color=cbblue](r4)--(rp5);
\draw [color=cbblue](r5)--(rp4);
\draw [color=cbblue](l4)--(lp5);
\draw [color=cbblue](l5)--(lp4);
\node (m) at (0, 0) {};
\node (mp) at (0, 2) {};
\draw [color=cbpurple](m)--(mp);
\draw [color=cbblue](r1)--(lp1);
\draw [color=cbblue](l1)--(rp1);
\draw [color=cbblue](r6)--(rp6);
\draw [color=cbblue](l6)--(lp6);
\end{scope}
\end{tikzpicture}
     \end{center}
    \caption{The vertex gadget of size 3.
    }
    \label{vg}
  \end{figure}
  We distinguish special areas in the vertex gadget, 
  in which we put other elements of our construction.
  These areas are called \emph{slots} and are marked with gray rectangles.
  We also number them as in the figure.
  The ones to the left of the purplish edge are called \emph{left slots}
  and the ones to the right -- \emph{right slots}.
  The vertex gadget of size $s$ has $s$ left slots and $s$ right slots.
  Furthermore, observe that there are only two ways to choose a perfect matching in this gadget:
  either take the blue edges and the purplish edge in the middle, or the yellowish edges and the purplish one.
  Choosing the blue (the yellowish) matching is interpreted
  as selecting (not selecting) the vertex to the cover
  and we say that the gadget is `selected' (`not selected').

  The second gadget, the \emph{edge gadget}, is created for every edge of $G$.
  It is shown in Figure~\ref{eg}.
  \begin{figure}[h]
    \begin{center}
      \begin{tikzpicture}
\begin{scope}[xscale=0.55, yscale=0.55]
\definecolor{cborange}{RGB}{230,159,0}
\definecolor{cbblue}{RGB}{0,114,178}
\definecolor{cbpurple}{RGB}{204,121,167}
\definecolor{cbyellow}{RGB}{240,228,66}
\definecolor{cbgreen}{RGB}{0,158,115}
\definecolor{cbvermillon}{RGB}{213,94,0}
\definecolor{mgray}{RGB}{220,220,220}
\definecolor{mlightblue}{RGB}{220,255,255}
\tikzstyle{every node}=[circle,minimum size=3pt,inner sep=0pt,draw,fill]
\tikzstyle{every path}=[line width = 1.9pt]
\scriptsize
\node [label=below:$b$] (b) at (0, 0) {};
\node [label=above:$a$] (a) at (0, 2) {};
\node [label=below:$c$] (c) at (2, 0) {};
\node [label=above:$d$] (d) at (10, 2) {};
\node [label=above:$e$] (e) at (12, 2) {};
\node [label=below:$f$] (f) at (12, 0) {};
\draw [color=cbgreen](a)--(b);
\draw [color=cbvermillon](d)--(b);
\draw [color=cbgreen](c)--(e);
\draw [color=cbvermillon](f)--(e);
\draw [color=cbvermillon](a)--(c);
\draw [color=cbgreen](d)--(f);
\end{scope}
\end{tikzpicture}
     \end{center}
    \caption{The edge gadget.}
    \label{eg}
  \end{figure}

  The construction proceeds as follows.
  First, for every $v \in V(G)$, we create 
  a copy of the vertex gadget of size $2d(v)$.
  We place them on the two horizontal lines in such a way
  that each gadget occupies a separate range of the $x$ axis, in any order.
  Now, for every edge $\set{u,v}$, where the gadget of $u$ is to the left of the gadget of $v$,
  we select two consecutive right slots in the gadget of $u$ and two consecutive left slots in the gadget of $v$, create a copy of the edge gadget and place its vertices as follows:
  \begin{itemize}
    \item vertices $a$ and $b$ in the left selected slot of the gadget of $u$,
    \item vertex $c$ in the right selected slot of the gadget of $u$,
    \item vertex $d$ in the left selected slot of the gadget of $v$,
    \item verticed $e$ and $f$ in the right selected slot of the gadget of $v$.
  \end{itemize}
  Such a selection of consecutive slots for each edge is of course possible,
  as we set the size of the vertex gadget to be $2d(v)$.
  See Figure~\ref{vgeg} for a complete example.
  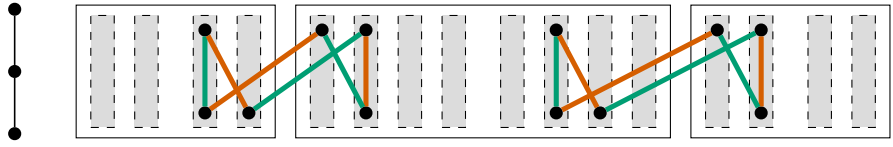
\begin{figure}[h]
    \begin{center}
      \begin{tikzpicture}
\begin{scope}[xscale=0.55, yscale=0.55]
\definecolor{cborange}{RGB}{230,159,0}
\definecolor{cbblue}{RGB}{0,114,178}
\definecolor{cbpurple}{RGB}{204,121,167}
\definecolor{cbyellow}{RGB}{240,228,66}
\definecolor{cbgreen}{RGB}{0,158,115}
\definecolor{cbvermillon}{RGB}{213,94,0}
\definecolor{mgray}{RGB}{220,220,220}
\definecolor{mlightblue}{RGB}{220,255,255}
\tikzstyle{every node}=[circle,minimum size=3pt,inner sep=0pt,draw,fill]
\tikzstyle{every path}=[line width = 1.9pt]
\scriptsize
\begin{scope}[xscale=0.7, yscale=1.0]
\node (x) at (-16, -0.5) {};
\node (y) at (-16, 1) {};
\node (z) at (-16, 2.5) {};
    \draw [line width = 0.6pt] (x) -- (y);
    \draw [line width = 0.6pt] (y) -- (z);
\filldraw[fill=white, draw=black, thin] (-13.9, -0.6) rectangle (-7.1, 2.6);
\filldraw[fill=mgray, draw=black,thin,dashed] (-9.9, -0.35) rectangle (-9.1, 2.35);
\filldraw[fill=mgray, draw=black,thin,dashed] (-11.1, -0.35) rectangle (-11.9, 2.35);
\filldraw[fill=mgray, draw=black,thin,dashed] (-8.4, -0.35) rectangle (-7.6, 2.35);
\filldraw[fill=mgray, draw=black,thin,dashed] (-12.6, -0.35) rectangle (-13.4, 2.35);
\filldraw[fill=white, draw=black, thin] (-6.4, -0.6) rectangle (6.4, 2.6);
\filldraw[fill=mgray, draw=black,thin,dashed] (0.6, -0.35) rectangle (1.4, 2.35);
\filldraw[fill=mgray, draw=black,thin,dashed] (-0.6, -0.35) rectangle (-1.4, 2.35);
\filldraw[fill=mgray, draw=black,thin,dashed] (2.1, -0.35) rectangle (2.9, 2.35);
\filldraw[fill=mgray, draw=black,thin,dashed] (-2.1, -0.35) rectangle (-2.9, 2.35);
\filldraw[fill=mgray, draw=black,thin,dashed] (3.6, -0.35) rectangle (4.4, 2.35);
\filldraw[fill=mgray, draw=black,thin,dashed] (-3.6, -0.35) rectangle (-4.4, 2.35);
\filldraw[fill=mgray, draw=black,thin,dashed] (5.1, -0.35) rectangle (5.9, 2.35);
\filldraw[fill=mgray, draw=black,thin,dashed] (-5.1, -0.35) rectangle (-5.9, 2.35);
\filldraw[fill=white, draw=black, thin] (7.1, -0.6) rectangle (13.9, 2.6);
\filldraw[fill=mgray, draw=black,thin,dashed] (11.1, -0.35) rectangle (11.9, 2.35);
\filldraw[fill=mgray, draw=black,thin,dashed] (9.9, -0.35) rectangle (9.1, 2.35);
\filldraw[fill=mgray, draw=black,thin,dashed] (12.6, -0.35) rectangle (13.4, 2.35);
\filldraw[fill=mgray, draw=black,thin,dashed] (8.4, -0.35) rectangle (7.6, 2.35);
\node (b) at (-9.5, 0) {};
\node (a) at (-9.5, 2) {};
\node (c) at (-8, 0) {};
\node (d) at (-5.5, 2) {};
\node (e) at (-4, 2) {};
\node (f) at (-4, 0) {};
\draw [color=cbgreen](a)--(b);
\draw [color=cbvermillon](d)--(b);
\draw [color=cbgreen](c)--(e);
\draw [color=cbvermillon](f)--(e);
\draw [color=cbvermillon](a)--(c);
\draw [color=cbgreen](d)--(f);
\node (b) at (2.5, 0) {};
\node (a) at (2.5, 2) {};
\node (c) at (4, 0) {};
\node (d) at (8, 2) {};
\node (e) at (9.5, 2) {};
\node (f) at (9.5, 0) {};
\draw [color=cbgreen](a)--(b);
\draw [color=cbvermillon](d)--(b);
\draw [color=cbgreen](c)--(e);
\draw [color=cbvermillon](f)--(e);
\draw [color=cbvermillon](a)--(c);
\draw [color=cbgreen](d)--(f);
\end{scope}
\end{scope}
\end{tikzpicture}
     \end{center}
    \caption{
     A graph $G$
     and a possible bipartite graph 
     obtained by 
     passing $G$ to the reduction algorithm,
     vertex gadgets presented schematically.
   }
    \label{vgeg}
  \end{figure}
  The edge gadget admits exactly two perfect matchings as well and 
  just like previously, we give interpretations to these matchings.
  If the reddish (greenish) matching is selected, 
  we say that the edge gadget is `covered'
  at the right (left) side
  and `not covered' at the left (right) side.
  Our naming convention may be confusing, as in the case of vertex covers,
  an edge may be covered at both sides,
  and our edge gadgets are always `not covered' at one side.
  The property that we want to enforce is as follows:
  in every optimal solution,
  when the edge gadget is `covered' at one side, 
  the corresponding vertex gadget must be `selected',
  and when the edge gadget is `not covered' at this side, the vertex gadget may be either `selected' or `not selected'.

Now, we assume that the positions of all the gadgets are fixed
and count the number of crossing edges.
In our analysis, we are only interested in how this number changes
when a different matching is chosen,
and for this reason we introduce constants $c_1, c_2, \ldots$
that are dependent on the way the gadgets were assembled on the two horizontal lines,
but not on the choice of matching.
First, we count such crossings, where an edge of the vertex gadget crosses
another edge of the same vertex gadget. As the vertex gadget of size $s$ admits $2s+1$
intersections if `selected' and $2s$ otherwise, this number is equal to:

\begin{displaymath}
  \cnt{s} + \sum_{v \in V(G)} 2 \cdot 2d(v) = \cnt{s} + c_1,
\end{displaymath}
where $\cnt{s}$ is the number of `selected' vertex gadgets.

The number of intersections between edges of edge gadgets
turns out to be independent of the matching chosen
and we denote it by $c_2$.
To see this, first observe that the number of intersections inside the edge gadget is always 1.
Second, note that for two different copies of the edge gadget,
the number of intersections between them is either 0, 1, 2 or 3,
but in all cases it it independent of the choice of the matching.

It remains to count the number of intersections such that one edge belongs to the vertex gadget
and the other to the edge gadget.
Fix $v \in V(G)$ and $e \in E(G)$.
We count intersections between edges of the vertex gadget of $v$ and edges of the edge gadget of $e$.
If $v \notin e$,
this number is independent of the choice of the matching.
Hence, we denote the number of such intersections
between every vertex gadget and every edge gadget by $c_3$.
Now assume that $v \in e$.
As the vertex gadget and the edge gadget admit 2 possible perfect matchings each,
we have 4 possibilities, as listed in Figure~\ref{cases}.
\begin{figure}[h]
  \begin{center}
    \begin{tikzpicture}
\begin{scope}[xscale=0.55, yscale=0.55]
\definecolor{cborange}{RGB}{230,159,0}
\definecolor{cbblue}{RGB}{0,114,178}
\definecolor{cbpurple}{RGB}{204,121,167}
\definecolor{cbyellow}{RGB}{240,228,66}
\definecolor{cbgreen}{RGB}{0,158,115}
\definecolor{cbvermillon}{RGB}{213,94,0}
\definecolor{mgray}{RGB}{220,220,220}
\definecolor{mlightblue}{RGB}{220,255,255}
\tikzstyle{every node}=[circle,minimum size=3pt,inner sep=0pt,draw,fill]
\tikzstyle{every path}=[line width = 1.9pt]
\scriptsize

\node (r1) at (1, 0) {};
\begin{scope}[shift={(0,0)}]
    \filldraw[fill=mgray, draw=black,thin,dashed] (1.8, -0.3) rectangle (2.2, 2.3);
    \filldraw[fill=mgray, draw=black,thin,dashed] (5.8, -0.3) rectangle (6.2, 2.3);
    \filldraw[fill=mgray, draw=black,thin,dashed] (9.8, -0.3) rectangle (10.2, 2.3);

    \node (r1) at (1, 0) {};
    \node (rp1) at (1, 2) {};
    \node (r2) at (3, 0) {};
    \node (rp2) at (3, 2) {};
    \node (r3) at (5, 0) {};
    \node (rp3) at (5, 2) {};
    \node (r4) at (7, 0) {};
    \node (rp4) at (7, 2) {};
    \node (r5) at (9, 0) {};
    \node (rp5) at (9, 2) {};
    \node (r6) at (11, 0) {};
    \node (rp6) at (11, 2) {};
    \draw [color=cbblue](r2)--(rp3);
    \draw [color=cbblue](r3)--(rp2);
    \draw [color=cbblue](r4)--(rp5);
    \draw [color=cbblue](r5)--(rp4);
    \draw [color=cbblue](r6)--(rp6);

    \node (b) at (2, 0) {};
    \node  (a) at (2, 2) {};
    \node  (c) at (6, 0) {};
    \node  (d) at (16, 2) {};
    \node  (e) at (18, 2) {};
    \node  (f) at (18, 0) {};
    \draw [color=cbgreen](a)--(b);
    \draw [color=cbgreen](c)--(e);
    \draw [color=cbgreen](d)--(f);

    \fill [color=white] (0, -0.3) rectangle (1.3, 2.3);
    \fill [color=white] (13, -0.3) rectangle (18.3, 2.3);

    \node [rectangle, fill=none, draw=none, align=left, anchor=west, text width = 4cm] at (14,1.1) {\normalsize vertex selected, \\ edge covered};

\end{scope}

\begin{scope}[shift={(0,-3.5)}]
    \filldraw[fill=mgray, draw=black,thin,dashed] (1.8, -0.3) rectangle (2.2, 2.3);
    \filldraw[fill=mgray, draw=black,thin,dashed] (5.8, -0.3) rectangle (6.2, 2.3);
    \filldraw[fill=mgray, draw=black,thin,dashed] (9.8, -0.3) rectangle (10.2, 2.3);

    \node (r1) at (1, 0) {};
    \node (rp1) at (1, 2) {};
    \node (r2) at (3, 0) {};
    \node (rp2) at (3, 2) {};
    \node (r3) at (5, 0) {};
    \node (rp3) at (5, 2) {};
    \node (r4) at (7, 0) {};
    \node (rp4) at (7, 2) {};
    \node (r5) at (9, 0) {};
    \node (rp5) at (9, 2) {};
    \node (r6) at (11, 0) {};
    \node (rp6) at (11, 2) {};
    \draw [color=cbblue](r2)--(rp3);
    \draw [color=cbblue](r3)--(rp2);
    \draw [color=cbblue](r4)--(rp5);
    \draw [color=cbblue](r5)--(rp4);
    \draw [color=cbblue](r6)--(rp6);

    \node (b) at (2, 0) {};
    \node  (a) at (2, 2) {};
    \node  (c) at (6, 0) {};
    \node  (d) at (16, 2) {};
    \node  (e) at (18, 2) {};
    \node  (f) at (18, 0) {};
    \draw [color=cbvermillon](f)--(e);
    \draw [color=cbvermillon](a)--(c);
    \draw [color=cbvermillon](d)--(b);

    \fill [color=white] (0, -0.3) rectangle (1.3, 2.3);
    \fill [color=white] (13, -0.3) rectangle (18.3, 2.3);

    \node [rectangle, fill=none, draw=none, align=left, anchor=west, text width = 4cm] at (14,1.1) {\normalsize vertex selected, \\ edge not covered};
\end{scope}

\begin{scope}[shift={(0,-7)}]
    \filldraw[fill=mgray, draw=black,thin,dashed] (1.8, -0.3) rectangle (2.2, 2.3);
    \filldraw[fill=mgray, draw=black,thin,dashed] (5.8, -0.3) rectangle (6.2, 2.3);
    \filldraw[fill=mgray, draw=black,thin,dashed] (9.8, -0.3) rectangle (10.2, 2.3);

    \node (r1) at (1, 0) {};
    \node (rp1) at (1, 2) {};
    \node (r2) at (3, 0) {};
    \node (rp2) at (3, 2) {};
    \node (r3) at (5, 0) {};
    \node (rp3) at (5, 2) {};
    \node (r4) at (7, 0) {};
    \node (rp4) at (7, 2) {};
    \node (r5) at (9, 0) {};
    \node (rp5) at (9, 2) {};
    \node (r6) at (11, 0) {};
    \node (rp6) at (11, 2) {};
    \draw [color=cborange](r1)--(rp2);
    \draw [color=cborange](rp1)--(r2);
    \draw [color=cborange](r3)--(rp4);
    \draw [color=cborange](rp3)--(r4);
    \draw [color=cborange](r5)--(rp6);
    \draw [color=cborange](rp5)--(r6);

    \node (b) at (2, 0) {};
    \node  (a) at (2, 2) {};
    \node  (c) at (6, 0) {};
    \node  (d) at (16, 2) {};
    \node  (e) at (18, 2) {};
    \node  (f) at (18, 0) {};
    \draw [color=cbgreen](a)--(b);
    \draw [color=cbgreen](c)--(e);
    \draw [color=cbgreen](d)--(f);

    \fill [color=white] (0, -0.3) rectangle (1.3, 2.3);
    \fill [color=white] (13, -0.3) rectangle (18.3, 2.3);

    \node [rectangle, fill=none, draw=none, align=left, anchor=west, text width = 4cm] at (14,1.1) {\normalsize vertex not selected, \\ edge covered};
\end{scope}

\begin{scope}[shift={(0,-10.5)}]
    \filldraw[fill=mgray, draw=black,thin,dashed] (1.8, -0.3) rectangle (2.2, 2.3);
    \filldraw[fill=mgray, draw=black,thin,dashed] (5.8, -0.3) rectangle (6.2, 2.3);
    \filldraw[fill=mgray, draw=black,thin,dashed] (9.8, -0.3) rectangle (10.2, 2.3);

    \node (r1) at (1, 0) {};
    \node (rp1) at (1, 2) {};
    \node (r2) at (3, 0) {};
    \node (rp2) at (3, 2) {};
    \node (r3) at (5, 0) {};
    \node (rp3) at (5, 2) {};
    \node (r4) at (7, 0) {};
    \node (rp4) at (7, 2) {};
    \node (r5) at (9, 0) {};
    \node (rp5) at (9, 2) {};
    \node (r6) at (11, 0) {};
    \node (rp6) at (11, 2) {};
    \draw [color=cborange](r1)--(rp2);
    \draw [color=cborange](rp1)--(r2);
    \draw [color=cborange](r3)--(rp4);
    \draw [color=cborange](rp3)--(r4);
    \draw [color=cborange](r5)--(rp6);
    \draw [color=cborange](rp5)--(r6);

    \node (b) at (2, 0) {};
    \node  (a) at (2, 2) {};
    \node  (c) at (6, 0) {};
    \node  (d) at (16, 2) {};
    \node  (e) at (18, 2) {};
    \node  (f) at (18, 0) {};
    \draw [color=cbvermillon](f)--(e);
    \draw [color=cbvermillon](a)--(c);
    \draw [color=cbvermillon](d)--(b);

    \fill [color=white] (0, -0.3) rectangle (1.3, 2.3);
    \fill [color=white] (13, -0.3) rectangle (18.3, 2.3);

    \node [rectangle, fill=none, draw=none, align=left, anchor=west, text width = 4cm] at (14,1.1) {\normalsize vertex not selected, \\ edge not covered};
\end{scope}

\end{scope}
\end{tikzpicture}
   \end{center}
  \caption{
    The 4 possible configurations of an intersection of the right part of the vertex gadget and the left part of the edge gadget.
  }
  \label{cases}
\end{figure}
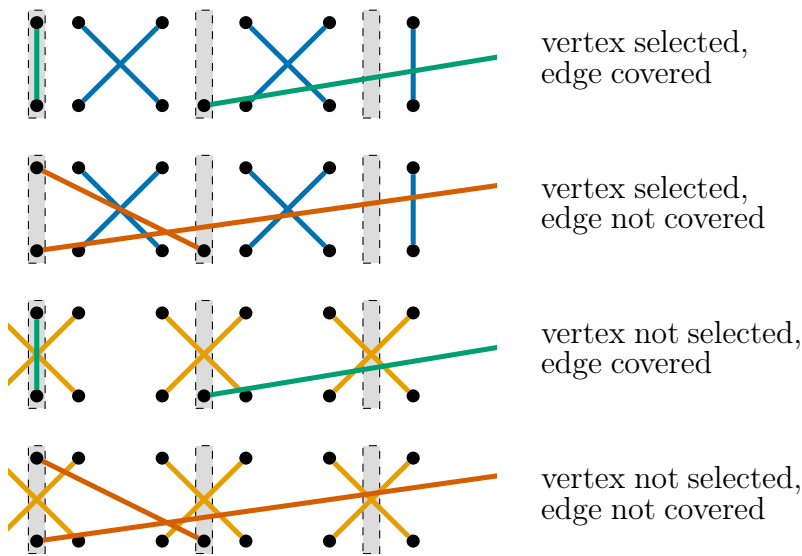
The figure does not lose generality: in the figure,
we are considering the right part of the vertex gadget
and the left part of the edge gadget,
but the analysis is the same in the opposite case.
Let $s, s + 1$ be the numbers of the two slots in the vertex gadget of $v$ 
occupied by vertices of the edge gadget of $e$.
The number of intersections between edges of the gadget of $v$ and edges of the gadget of $e$ is equal to:
\begin{itemize}
  \item $2(s-1) + 1 = 1 + c_4$ in the `vertex selected, edge covered' case,
  \item $2(s-1) + 5 = 5 + c_4$ in the `vertex selected, edge not covered' case,
  \item $2(s-1) + 3 = 3 + c_4$ in the `vertex not selected, edge covered' case,
  \item $2(s-1) + 5 = 5 + c_4$ in the `vertex not selected, edge not covered' case.
\end{itemize}

Let the variables $\cnt{sc}, \cnt{snc}, \cnt{nsc}, \cnt{nsnc}$
count occurences of each of the four cases above in the entire graph $H$, respectively.
The total number of crossing edges is equal to:
\begin{multline*}
  \cnt{s} + c_1 + c_2 + c_3  \\
  + (1 + c_4)\cnt{sc} 
  + (3 + c_4)\cnt{nsc}
  + (5 + c_4)\cnt{snc} 
  + (5 + c_4)\cnt{nsnc},
\end{multline*}
However, as every edge gadget is `covered' at one side and `not covered' at the other,
we have $\cnt{sc} + \cnt{nsc} = \cnt{snc} + \cnt{nsnc} = |E(G)|$
and hence the calculation simplifies to
\begin{multline*}
  \cnt{s} + c_1 + c_2 
  + 2 \cdot \cnt{nsc} 
  + (1 + c_4)|E(G)| + (5 + c_4)|E(G)| \\
  = \cnt{s} + 2 \cdot \cnt{nsc} + c_5.
\end{multline*}

To complete the description of the reduction algorithm, 
we set $m = k + c_5$.

It is straightforward to implement the reduction algorithm in polynomial time.
It remains to prove that $G$ admits a vertex cover of size $k$
if and only if $H$ admits a perfect matching with at most $m$ intersections.

First suppose that $G$ admits a vertex cover $C$ of size at most $k$.
Then one can choose the `selected' perfect matching for vertex gadgets
of every vertex in $C$
and the `not selected' perfect matching for every other vertex.
Moreover, as every edge of $G$ is covered,
one can choose perfect matchings in edge gadgets
so that their `covered' side is in a `selected' vertex gadget.
Then $\cnt{nsc} = 0$ and the number of intersecting edges is equal to
$\cnt{s} + c_5 \leq k + c_5 = m$,
so $H$ admits a perfect matching with at most $m$ intersections.

For the second implication, assume that $H$
admits a perfect matching with at most $m$ intersections.
Let $M$ be any matching with minimal number of intersections.
Observe that $\cnt{nsc} = 0$,
as if there exists a vertex gadget that is `not selected' and 
intersects a `covered' edge gadget,
one can choose the vertex gadget to be `selected' instead,
and achieve a perfect matching with fewer intersections,
which contradicts the minimality of $M$.
Now we construct a vertex cover of $G$: we select exactly the vertices
whose vertex gadgets were `selected'.
To see that this is a vertex cover, fix an edge of $G$.
At the `covered' side of its edge gadget, 
the vertex gadget is `selected', because $\cnt{nsc} = 0$.
Thus, the corresponding vertex is selected to the cover.
Finally, as the number of intersections in our construction is equal to $\cnt{s} + c_5$ 
and is at most $m$,
the size of the vertex cover, equal to $\cnt{s}$, 
is at most $m - c_5 = k$.
\end{proof}

Observe that in the proof above,
the size of the \textsc{Crossing-avoiding Matching}
instance is linear in the size of the $\textsc{Vertex Cover}$ instance.
Indeed, for every vertex $v \in G$ we produce $16d(v) + 2$ vertices of $H$,
and for every edge of $G$, six vertices are produced.
Hence, the reduction algorithm yields $2|V(G)| + 32|E(G)|$ vertices,
and, as as the vertices in $H$ are of degree at most 2,
at most that many edges.
Suppose there exists a $2^{o(|V(H)| + |E(H)|)}$ algorithm for \textsc{Crossing-avoiding Matching}.
Composing the reduction with such an algorithm
gives a $2^{o(|V(G)| + |E(G)|)}$ algorithm for $\textsc{Vertex Cover}$,
which contradicts the Exponential Time Hypothesis
according to Theorem~14.6.~ from the book by Cygan et~al.~\cite{CyganFKLMPPS15}.
This yields the following:

\begin{theorem}\label{t2}
  There is no $2^{o(|V(H)| + |E(H)|)}$ algorithm for \textsc{Crossing-avoiding Matching},
  unless ETH fails.
\end{theorem}

\section{Acknowledgements}

I would like to thank Grzegorz Gutowski for 
suggesting a direction for developing the gadgets,
and Paweł Rzążewski for introducing me to this problem,
pointing out the lower bound based on ETH,
and other helpful comments.

\bibliographystyle{plain}
\bibliography{paper}

\end{document}